\documentclass[leqno]{amsart}
\usepackage{amsmath,amsthm,amsfonts,amssymb}

\usepackage[margin=1.15in]{geometry}

\usepackage{graphicx}

\numberwithin{equation}{section}
\theoremstyle{plain}

\newtheorem{theorem}{Theorem}

\newtheorem{remark}[theorem]{Remark}

\begin{document}

\title[A Delayed Dual Risk Model]{A Delayed Dual Risk Model}

\author{Lingjiong Zhu}
\address
{Department of Mathematics \newline
\indent Florida State University \newline
\indent 1017 Academic Way \newline
\indent Tallahassee, FL-32306 \newline
\indent United States of America}
\email{
zhu@math.fsu.edu}

\date{4 September 2016.}

\begin{abstract}
In this paper, we study a dual risk model with delays in the spirit of Dassios-Zhao. When a new innovation 
occurs, there is a delay before the innovation turns into a profit. We obtain
large initial surplus asymptotics for the ruin probability and ruin time distributions. 
For some special cases, we get closed-form formulas. Numerical illustrations
will also be provided.
\end{abstract}

\maketitle

\section{Introduction}

The classic risk model is based on the surplus process $X_{t}=x+\rho t-\sum_{i=1}^{N_{t}}Y_{i}$,
where the insurer starts with the initial reserve $x$ and receives the premium at a constant rate $\rho$ 
and $Y_{i}$ are the claims. In recent years, 
a dual risk model has attracted substantial attention, 
in which the surplus process is modeled as
\begin{equation}\label{DualModel}
dX_{t}=-\rho dt+dJ_{t},\qquad X_{0}=x>0,
\end{equation}
where $\rho>0$ is the cost of running the company and $J_{t}=\sum_{i=1}^{N_{t}}Y_{i}$, 
is the stream of profits, where $Y_{i}$ are i.i.d. $\mathbb{R}^{+}$ valued random variables
with common probability density function $p(y)$, $y>0$
and they denote the profits from the innovations and
sometimes are referred to as the innovation sizes in the insurance
and finance literature and $N_{t}$
is a Poisson process with intensity $\lambda>0$. 
The dual risk model is used to model the wealth of a venture capital, an oil company, see e.g. \cite{Avanzi} 
or any business with random gains, see e.g. \cite{CheungI}.

Let $\tau$ be the first time that the wealth process
hits zero, that is,
\begin{equation}\label{ClassicalRuin}
\tau:=\inf\{t>0:X_{t}\leq 0\}.
\end{equation}
Let us assume the net condition holds:
\begin{equation}
\lambda\mathbb{E}[Y_{1}]>\rho.
\end{equation}
The net condition will always be assumed throughout the paper unless specified otherwise.
The infinite horizon ruin probability is given by, see e.g. Avanzi et al. \cite{Avanzi}
\begin{equation}
\mathbb{P}(\tau<\infty|X_{0}=x)=e^{-\alpha x},
\end{equation}
where $\alpha$ is the unique positive value that satisfies
the equation:
\begin{equation}\label{alphaEqn}
\alpha\rho+\lambda\int_{0}^{\infty}[e^{-\alpha y}-1]p(y)dy=0.
\end{equation}

The classical risk model with delays has been studied in the insurance literature
since Waters and Papatriandafylou \cite{Waters}.
Claims could have already occurred but have not been settled or reported
immediately. The delay can be caused by many different reasons, e.g.
IBNR (Incurred But Not Reported) and IBNR (Reported But Not Settled).
A discrete-time model for a risk process allowing claims being delayed
was considered in Waters and Papatriandafylou \cite{Waters}
and Trufin et al. \cite{Trufin}. Boogaert and Haezendonck \cite{Boogaert}
studied a liability process with delays in the framework of economical environment. 
Yuen et al. \cite{Yuen} introduced a continuous time model with one claim settled
immediately and the other claim, named by-claim, settled with delay for the each time
of claim occurrences. Dassios and Zhao \cite{DassiosZhao}
studied a classical risk model with delays, in which the claim sizes have light-tailed distributions
and each of the claims will be settled in a randomly delayed period of time. 
They used the connection to the non-homogenous Poisson process to obtain
the asymptotic expressions for the ruin probability and a finer asympotitc formula
is obtained for exponentially delayed claims and an exact formula is derived
when the claims are also exponentially distributed.
Delaying claims have also been modeled
using the Poisson shot noise processes or a Cox process with shot noise intensity, 
see e.g. Kl\"{u}ppelberg and Mikosch \cite{Kluppelberg}, Br\'{e}maud \cite{Bremaud}, Macci and Torrisi \cite{Macci}
and Albrecher and Asmussen \cite{AA}. 

We are interested to study a delayed dual risk model. 
Our model is analogous to the delayed-claim model for the classical
risk model studied by Dassios and Zhao \cite{DassiosZhao}. The dual risk model
can be used to model a venture capital, high tech company, oil company or any business
with random growth \cite{CheungI}. In \eqref{DualModel}, $N_{t}$ is the arrival
times of the innovations, which are also taken as the arrival times of the profits. 
In the standard dual risk model, 
the profits arrive as soon as the innovations occur. 
In reality, there is usually a period of delay before
the innovation can be turned into a profit. 
For example, there can be a delay caused by patent application.
The present average wait time until the USPTO (United States Patent and Trademark Office) 
provides the results of the Patent Examiner's first substantive review and examination 
(average pendency to first office action) of the patent application is about 17 months. 
The average time it takes to obtain a patent from the patent office at this time is about 27 months. 
See the current wait time statistics at the USPTO website 
\footnote{\texttt{http://www.uspto.gov/dashboards/patents/main.dashxml}}.
Another example is the exploration for the oil companies.
After the discovery of a new oil field, it takes time for the oil company to make
profits from the discovery. Time is needed for activities such as building
surface infrastructure and pipelines etc., before the production can start, which causes a delay.
For instance, in US, approval of an interstate pipeline takes an average time of 15 months
and in addition, the pipeline construction usually takes 6 to 18 months.
See the website of US Energy Information Administration 
\footnote{\texttt{http://www.eia.gov/pub/oil$_{-}$gas/natural$_{-}$gas/analysis$_{-}$publications
/ngpipeline/develop.html}}.
The delay can also caused by the regulations. 
The influence of regulations on innovation has been well studied 
in the economics literature, see e.g. \cite{Blind}.
The regulatory delay exists when the regular does not allow
the introduction of new products without regulatory review and approval. 
The regulatory delay is the time between the firm's submission of a new product to the regulator
for approval and the granting of approval, see e.g. \cite{Prieger}.
The regulated firms in the telecommunications, pharmaceutical, banking and other industries
often claim that regulatory delays are long and costly. The data for regulatory delays
for several states in US can be found in \cite{Prieger}.
The regulatory delays can have an effect on the firm's innovations.
The longer the regulatory delay, the less a regulated firm will be willing
to spend in R\&D on innovation, see \cite{Braeutigam}.
Regulatory delay exerts a multiplier effect on total time to market, 
because when the firm expects the regulator to take longer to grant approval, 
the firm delays its product introduction, see e.g. \cite{PriegerII}.
All these examples discussed above motivate us to study a dual risk model with delays.

Before we proceed, let us first have a brief review of the dual risk models
in the finance and insurance literature. 
Avanzi et al. \cite{Avanzi} studied the optimal dividend problem, i.e., De Finetti problem \cite{DeFinetti},
for the dual risk model. The optimal strategy is given by a barrier strategy
and more explicit formula was obtained when the innovation sizes are 
exponentially distributed. Afonso et al. \cite{Afonso} established a connection
between the dual and classical risk models, presented a new approach
for computing the expected discounted dividends and derived some known
and also some new results. Cheung and Drekic \cite{CheungII} studied
the dividend moments in the dual risk model.
Ng \cite{Ng} also studied the optimal dividend problem in the dual risk model.
Unlike \cite{Avanzi}, Ng \cite{Ng} considered the case when the surplus is above
the threshold, the dividend is paid out at a constant rate to the shareholders
while in \cite{Avanzi} the entire surplus above the threshold is paid out immediately
to the shareholders as the dividend. Ng \cite{Ng} derived a set of two integro-differential
equations satisfied by the expected total discounted dividends until ruin and showed
the equations can be solved by using only one of the two integro-differential equations.
Albrecher et al. \cite{Albrecher}
introduced tax payments to the dual risk model. The ruin probability, Laplace transform
of the ruin time, moments of the discounted tax payments were computed. The critical surplus
level at which it is optimal to collect tax payments was also determined.
Ng \cite{NgII} studied the dual risk model when the random gains follow a phase-type distribution. 
Two pairs of upcrossing and downcrossing barrier probabilities were derived.
Avanzi et al. \cite{AvanziII} considered the dividends in a dual risk model in which
dividend decisions are made periodically, but ruin can still occur at any time. 
There are also works on the generalization of the Poisson dual risk model.
For example, Rodr\'{i}guez et al. \cite{RCE} studied the Erlang(n) dual risk model,
in which they assumed that the waiting times are Erlang(n) distributed, instead of exponentially
distributed. They obtained ruin probability and the Laplace transform of ruin time using
the roots of the fundamental and the generalized Lundberg's equation. Erlang(n) dual risk model
was also used in Avanzi et al. \cite{AvanziII}.
Yang and Sendova \cite{YS} studied the properties of the ruin time for the Sparre-Andersen dual model.
Zhu \cite{Zhu} studied state-dependent dual risk model, in which both the running cost
and the arrival of future profits are dependent on the surplus level. 
Bayraktar and Egami \cite{BE} studied venture capital investment using a dual risk model.
Recently, Fahim and Zhu \cite{FZ} studied the optimal investment in research
and investment to boost the future profits using a dual risk model. 
Other than the classical ruin time, Parisian ruin times have also been studied recently
for the dual risk models, see e.g. \cite{YSII}. 

Now let us introduce the model.
Consider the dual risk model \eqref{DualModel}.
Let us introduce the random delays in the spirit of Dassios and Zhao \cite{DassiosZhao}.
Let $N_{t}$ be the number of realized profits within the time interval $[0,t]$ and assume $N_{0}=0$.
$\{T_{k}\}_{k=1}^{\infty}$, $\{L_{k}\}_{k=1}^{\infty}$
denote the random times of the innovation arrival and the corresponding delays.
Therefore, $\{T_{k}+L_{k}\}_{k=1}^{\infty}$ are the random times
when the profits are realized. 
Thus,
\begin{equation}
N_{t}=\sum_{k=1}^{\infty}1_{T_{k}+L_{k}\leq t}.
\end{equation}

We assume that $\{T_{k}-T_{k-1}\}_{k\geq 1}$ (with $T_{0}:=0$) are i.i.d. exponentially distributed 
with parameter $\lambda>0$ so that
without delay, the arrival process is a Poisson process with parameter $\lambda>0$. 
$L_{k}$ are assumed to be independent and identically distributed non-negative
random variables with cumulative distribution function $L(t)$ and we define $\bar{L}(t):=1-L(t)$. 

The ruin time after time $t\geq 0$ is defined as
\begin{equation}
\tau_{t}:=
\inf\{s:s>t, X_{s}\leq 0\},
\end{equation}
and $\tau_{t}=\infty$ if the ruin never occurs. 

We are interested in the ultimate ruin probability at time $t$, that is,
\begin{equation}
\psi(x,t):=\mathbb{P}(\tau_{t}<\infty|X_{t}=x),
\end{equation}
or equivalently the ultimate survival probability at time $t$, that is,
\begin{equation}
\phi(x,t):=1-\psi(x,t).
\end{equation}

The key observation is that a delayed (or displaced) Poisson process
is still a (non-homogeneous) Poisson process, see e.g. Mirasol \cite{Mirasol}, 
and also see Dassios and Zhao \cite{DassiosZhao} and more precisely,
$N_{t}$ is a non-homogeneous Poisson process
with intensity $\lambda L(t)$ at time $t$.

Therefore, similarly as in Dassios and Zhao \cite{DassiosZhao}, the ultimate survival probability $\phi(x,t)$ satisfies the equation:
\begin{equation}\label{survivalEqn}
\frac{\partial\phi}{\partial t}
-\rho\frac{\partial\phi}{\partial x}
+\lambda L(t)\int_{0}^{\infty}[\phi(x+y,t)-\phi(x,t)]p(y)dy=0,
\end{equation}
with the boundary condition $\phi(0,t)\equiv 0$. 
Similarly, one can also write down the equation
for the Laplace transform of the ruin time.

In this paper, we are interested to study the ultimate ruin probability, 
the Laplace transform of the ruin time, and the probability density
function of the ruin time. Asymptotics for 
the large initial surplus will be obtained. We will also 
discuss some exactly solvable examples for which
the ultimate ruin probability, Laplace transform of the ruin time,
and the probability density function of the ruin time have
closed-form formulas. Finally, some numerical illustrations will be provided.


\section{Large Initial Surplus Asymptotics}

The ultimate survival probability satisfies the equation \eqref{survivalEqn}
which is challenging to solve in closed-form. Instead, we give 
non-trivial lower and upper bound estimates for the ruin probabilities, ruin time
and provide large initial surplus asymptotics. 

\begin{theorem}\label{ThmI}
Assume that $\int_{0}^{\infty}\bar{L}(t)dt<\infty$.
We have the following estimate for the ruin probability:
\begin{equation}
e^{-\alpha x}e^{\alpha\rho\int_{t}^{\infty}\bar{L}(s)ds}
e^{-\alpha\rho\int_{\frac{x}{\rho}+t}^{\infty}\bar{L}(s)ds}
\leq\mathbb{P}(\tau_{t}<\infty|X_{t}=x)
\leq e^{-\alpha x}e^{\alpha\rho\int_{t}^{\infty}\bar{L}(s)ds},
\end{equation}
where $\alpha$ is the unique positive value that satisfies \eqref{alphaEqn}.
In particular, as $x\rightarrow\infty$, we have
\begin{equation}
\mathbb{P}(\tau_{t}<\infty|X_{t}=x)
\sim e^{-\alpha x}e^{\alpha\rho\int_{t}^{\infty}\bar{L}(s)ds}.
\end{equation}
\end{theorem}

\begin{proof}
Consider $w(x,t)=(1-e^{-\alpha x})e^{-c\int_{t}^{\infty}\bar{L}(s)ds}$.
Then, 
\begin{align}
&\frac{\partial w}{\partial t}
-\rho\frac{\partial w}{\partial x}
+\lambda L(t)\int_{0}^{\infty}[w(x+y,t)-w(x,t)]p(y)dy
\\
&=c\bar{L}(t)(1-e^{-\alpha x})e^{-c\int_{t}^{\infty}\bar{L}(s)ds}
+\lambda\bar{L}(t)e^{-c\int_{t}^{\infty}\bar{L}(s)ds}
e^{-\alpha x}\int_{0}^{\infty}[e^{-\alpha y}-1]p(y)dy
\nonumber
\\
&=c\bar{L}(t)e^{-c\int_{t}^{\infty}\bar{L}(s)ds},
\nonumber
\end{align}
if we choose
\begin{equation}
c=-\lambda\int_{0}^{\infty}[1-e^{-\alpha y}]p(y)dy=-\alpha\rho.
\end{equation}
It is clear that $w(x,t)\in C^{1,1}_{b}$, that is, bounded and continuously differentiable in both $x$ and $t$. 
Hence, by It\^{o}'s formula,
\begin{equation}
\mathbb{E}_{X_{t}=x}[w(X_{\tau_{t}},\tau_{t})]=w(x,t)-\mathbb{E}_{X_{t}=x}\left[\int_{t}^{\tau_{t}}\alpha\rho
\bar{L}(s)e^{\alpha\rho\int_{s}^{\infty}\bar{L}(u)du}ds\right].
\end{equation}
Notice that $X_{\tau_{t}}=0$ if $\tau_{t}<\infty$ and otherwise it is infinity. 
Therefore, 
\begin{equation}
\mathbb{E}_{X_{t}=x}[w(X_{\tau_{t}},\tau_{t})]
=\mathbb{E}_{X_{t}=x}\left[(1-e^{-\alpha X_{\tau_{t}}})e^{\alpha\rho\int_{\tau_{t}}^{\infty}\bar{L}(s)ds}\right]
=\mathbb{P}(\tau_{t}=\infty|X_{t}=x).
\end{equation}
Thus, we have
\begin{equation}\label{EqnI}
\mathbb{P}(\tau_{t}=\infty|X_{t}=x)-(1-e^{-\alpha x})e^{\alpha\rho\int_{t}^{\infty}\bar{L}(s)ds}
=-\mathbb{E}_{X_{t}=x}\left[\int_{t}^{\tau_{t}}\alpha\rho
\bar{L}(s)e^{\alpha\rho\int_{s}^{\infty}\bar{L}(u)du}ds\right].
\end{equation}
It is clear that
\begin{equation}
\mathbb{E}_{X_{t}=x}\left[\int_{t}^{\tau_{t}}\alpha\rho
\bar{L}(s)e^{\alpha\rho\int_{s}^{\infty}\bar{L}(u)du}ds\right]
\leq\int_{t}^{\infty}\alpha\rho
\bar{L}(s)e^{\alpha\rho\int_{s}^{\infty}\bar{L}(u)du}ds.
\end{equation}

On the other hand, we have
\begin{align}
&\mathbb{E}_{X_{t}=x}\left[\int_{t}^{\tau_{t}}\alpha\rho
\bar{L}(s)e^{\alpha\rho\int_{s}^{\infty}\bar{L}(u)du}ds\right]
\\
&\geq\mathbb{P}(\tau_{t}=\infty|X_{t}=x)\int_{t}^{\infty}\alpha\rho
\bar{L}(s)e^{\alpha\rho\int_{s}^{\infty}\bar{L}(u)du}ds
\nonumber
\\
&\qquad\qquad\qquad
+\mathbb{E}_{X_{t}=x}\left[\int_{t}^{\tau_{t}}\alpha\rho
\bar{L}(s)e^{\alpha\rho\int_{s}^{\infty}\bar{L}(u)du}ds\cdot 1_{\tau_{t}<\infty}\right]
\nonumber
\\
&\geq
\mathbb{P}(\tau_{t}=\infty|X_{t}=x)\int_{t}^{\infty}\alpha\rho
\bar{L}(s)e^{\alpha\rho\int_{s}^{\infty}\bar{L}(u)du}ds
\nonumber
\\
&\qquad\qquad\qquad
+\int_{t}^{\frac{x}{\rho}+t}\alpha\rho
\bar{L}(s)e^{\alpha\rho\int_{s}^{\infty}\bar{L}(u)du}ds\mathbb{P}(\tau_{t}<\infty|X_{t}=x),
\nonumber
\end{align}
where we used the fact that $\tau_{t}\geq\frac{x}{\rho}+t$ a.s.

We can compute that
\begin{align}\label{EqnIII}
\int_{t}^{\infty}\alpha\rho
\bar{L}(s)e^{\alpha\rho\int_{s}^{\infty}\bar{L}(u)du}ds
&=-\int_{t}^{\infty}e^{\alpha\rho\int_{s}^{\infty}\bar{L}(u)du}d\left(\alpha\rho\int_{s}^{\infty}\bar{L}(u)du\right)
\\
&=e^{\alpha\rho\int_{t}^{\infty}\bar{L}(s)ds}-1.
\nonumber
\end{align}
Therefore,
\begin{align}\label{EqnII}
&\mathbb{E}_{X_{t}=x}\left[\int_{t}^{\tau_{t}}\alpha\rho
\bar{L}(s)e^{\alpha\rho\int_{s}^{\infty}\bar{L}(u)du}ds\right]
\\
&\geq
\int_{t}^{\infty}\alpha\rho
\bar{L}(s)e^{\alpha\rho\int_{s}^{\infty}\bar{L}(u)du}ds
-\int_{\frac{x}{\rho}+t}^{\infty}\alpha\rho
\bar{L}(s)e^{\alpha\rho\int_{s}^{\infty}\bar{L}(u)du}ds\mathbb{P}(\tau_{t}<\infty|X_{t}=x)
\nonumber
\\
&=e^{\alpha\rho\int_{t}^{\infty}\bar{L}(s)ds}-1-\left(e^{\alpha\rho\int_{\frac{x}{\rho}+t}^{\infty}\bar{L}(s)ds}-1\right)
\mathbb{P}(\tau_{t}<\infty|X_{t}=x)
\nonumber
\end{align}
Hence, from \eqref{EqnI} and \eqref{EqnII}, we have
\begin{align}
&\mathbb{P}(\tau_{t}=\infty|X_{t}=x)-(1-e^{-\alpha x})e^{\alpha\rho\int_{t}^{\infty}\bar{L}(s)ds}
\\
&\leq 1-e^{\alpha\rho\int_{t}^{\infty}\bar{L}(s)ds}-\left(1-e^{\alpha\rho\int_{\frac{x}{\rho}+t}^{\infty}\bar{L}(s)ds}\right)
\mathbb{P}(\tau_{t}<\infty|X_{t}=x)
\nonumber
\\
&=e^{\alpha\rho\int_{\frac{x}{\rho}+t}^{\infty}\bar{L}(s)ds}
+\left(1-e^{\alpha\rho\int_{\frac{x}{\rho}+t}^{\infty}\bar{L}(s)ds}\right)\mathbb{P}(\tau_{t}=\infty|X_{t}=x)
-e^{\alpha\rho\int_{t}^{\infty}\bar{L}(s)ds}
\nonumber
\end{align}
which implies that
\begin{equation}
\mathbb{P}(\tau_{t}=\infty|X_{t}=x)
\leq 1-e^{-\alpha x}e^{\alpha\rho\int_{t}^{\infty}\bar{L}(s)ds}e^{-\alpha\rho\int_{\frac{x}{\rho}+t}^{\infty}\bar{L}(s)ds}.
\end{equation}

On the other hand, from \eqref{EqnI} and \eqref{EqnIII}, we have 
\begin{equation}
\mathbb{P}(\tau_{t}=\infty|X_{t}=x)
\geq(1-e^{-\alpha x})e^{\alpha\rho\int_{t}^{\infty}\bar{L}(s)ds}
+1-e^{\alpha\rho\int_{t}^{\infty}\bar{L}(s)ds}
=1-e^{-\alpha x}e^{\alpha\rho\int_{t}^{\infty}\bar{L}(s)ds}.
\end{equation}
Hence, we conclude that
\begin{equation}
e^{-\alpha x}e^{\alpha\rho\int_{t}^{\infty}\bar{L}(s)ds}
e^{-\alpha\rho\int_{\frac{x}{\rho}+t}^{\infty}\bar{L}(s)ds}
\leq\mathbb{P}(\tau_{t}<\infty|X_{t}=x)
\leq e^{-\alpha x}e^{\alpha\rho\int_{t}^{\infty}\bar{L}(s)ds}.
\end{equation}
\end{proof}

\begin{remark}
In Dassios and Zhao \cite{DassiosZhao}, they obtained the large initial
surplus asymptotics for the ruin probabilities for classical risk model
with delays by taking the Laplace transform of the integro-differential equation
with respect to the space variable and analyzing the equation. Finer
estimates are obtained for exponentially delayed claims and explicit formulas
are derived when the claims also also exponentially distributed. It will be
interesting to see if similar techniques can be applied to study the dual risk model
with delays.
\end{remark}

Next, let us give some examples
for which the asymptotics in Theorem \ref{ThmI} becomes explicit.

\begin{remark}
Let us assume that $p(y)=\nu e^{-\nu y}$ for some $\gamma>0$.
Then from \eqref{alphaEqn}, we get $\alpha=\frac{\lambda}{\rho}-\nu$.

(i) (sub-exponential delay) Assume that $\bar{L}(t)=\frac{1}{(1+t)^{\gamma}}$ where $\gamma>1$. Then
\begin{equation}
\mathbb{P}(\tau_{t}<\infty|X_{t}=x)\sim e^{-(\frac{\lambda}{\rho}-\nu)x
+(\lambda-\nu\rho)\frac{1}{\gamma-1}\frac{1}{(1+t)^{\gamma-1}}},\qquad\text{as $x\rightarrow\infty$}.
\end{equation}

(ii) (exponential delay) Assume that $\bar{L}(t)=e^{-\gamma t}$ where $\gamma>0$. Then
\begin{equation}
\mathbb{P}(\tau_{t}<\infty|X_{t}=x)\sim e^{-(\frac{\lambda}{\rho}-\nu)x
+(\lambda-\nu\rho)\frac{1}{\gamma}e^{-\gamma t}},\qquad\text{as $x\rightarrow\infty$}.
\end{equation}

(iii) (super-exponential delay) Assume that $\bar{L}(t)=e^{-\gamma t^{2}}$ where $\gamma>0$. Then
\begin{equation}
\mathbb{P}(\tau_{t}<\infty|X_{t}=x)\sim e^{-(\frac{\lambda}{\rho}-\nu)x
+(\lambda-\nu\rho)\sqrt{\frac{\pi}{\gamma}}(1-N(\sqrt{2\gamma}t))},\qquad\text{as $x\rightarrow\infty$},
\end{equation} 
where $N(\cdot)$ is the cumulative distribution function of a standard normal random variable
with mean $0$ and variance $1$.
\end{remark}

Next, let us obtain
an asymptotic estimate for the probability density function of $\tau_{t}$
for large initial surplus $x$ limit. Before we proceed, let us 
point out that the probability density function for the ruin time for the dual risk model (without delays)
was first obtained by Prabhu \cite{Prabhu}.

\begin{theorem}\label{ThmII}
Assume that $\int_{0}^{\infty}\bar{L}(t)dt<\infty$.
We have the following estimate for the Laplace transform of the ruin time $\tau_{t}$, as $x\rightarrow\infty$,
\begin{equation}
\mathbb{E}_{X_{t}=x}\left[e^{-\theta\tau_{t}}\right]\sim
e^{-\beta x}e^{(\beta\rho-\theta)\int_{t}^{\infty}\bar{L}(s)ds-\theta t},
\end{equation}
where $\beta>0$ is the unique positive value that satisfies the equation:
\begin{equation}\label{betaEqn}
\rho\beta+\lambda\int_{0}^{\infty}[e^{-\beta y}-1]p(y)dy-\theta=0.
\end{equation}
As a result, we have the following asymptotic result for the probability density function of the ruin time $\tau_{t}$
at any finite time $T>\frac{x}{\rho}+t$ for $x\rightarrow\infty$:
\begin{align}
f(T;t,x)&\sim
\left(x-\rho\int_{t}^{\infty}\bar{L}(s)ds\right)\sum_{n=1}^{\infty}\frac{(\lambda/\rho)^{n}}{n!}
\left(\rho(T-t)-\rho\int_{t}^{\infty}\bar{L}(s)ds\right)^{n-1}
\\
&\qquad\qquad\qquad\cdot
e^{-\lambda(T-t)+\lambda\int_{t}^{\infty}\bar{L}(s)ds}
p^{\ast n}(\rho(T-t)-x),
\nonumber
\end{align}
where $p^{\ast n}(\cdot)$ is the probability density function
of $Y_{1}+Y_{2}+\cdots+Y_{n}$.
\end{theorem}

\begin{proof}
Define the infinitesimal generator as
\begin{equation}
\mathcal{L}_{t}w
=\frac{\partial w}{\partial t}
-\rho\frac{\partial w}{\partial x}
+\lambda L(t)\int_{0}^{\infty}[w(x+y,t)-w(x,t)]p(y)dy.
\end{equation}
Take $w(x,t)=e^{-\beta x-\theta t-c\int_{t}^{\infty}\bar{L}(s)ds}$,
where
\begin{equation}
c=-\lambda\int_{0}^{\infty}[1-e^{-\beta y}]p(y)dy
=\theta-\beta\rho.
\end{equation}
Then, we can compute that
\begin{align*}
\mathcal{L}_{t}w
&=(-\theta+c\bar{L}(t))w
+\beta\rho w+\lambda L(t)\int_{0}^{\infty}[e^{-\beta y}-1]p(y)dy\cdot w
\\
&=-cw+c\bar{L}(t)w+cL(t)w=0.
\end{align*}
By It\^{o}'s formula,
\begin{equation}
\mathbb{E}_{X_{t}=x}[w(X_{\tau_{t}},t)]
=w(x,t).
\end{equation}
Note that when $\tau_{t}<\infty$, $X_{\tau_{t}}=0$. 
Therefore, we get
\begin{equation}
\mathbb{E}_{X_{t}=x}\left[e^{-\theta\tau_{t}-c\int_{\tau_{t}}^{\infty}\bar{L}(s)ds}\right]
=e^{-\beta x-\theta t-c\int_{t}^{\infty}\bar{L}(s)ds}.
\end{equation}
In other words,
\begin{equation}
\mathbb{E}_{X_{t}=x}\left[e^{-\theta(\tau_{t}-t)+(\beta\rho-\theta)\int_{\tau_{t}}^{\infty}\bar{L}(s)ds}\right]
=e^{-\beta x+(\beta\rho-\theta)\int_{t}^{\infty}\bar{L}(s)ds}.
\end{equation}
Note that $\beta\rho-\theta=-c\geq 0$, and 
$\tau_{t}\geq t+\frac{x}{\rho}$, 
hence, we conclude that
\begin{equation}
e^{-\beta x+(\beta\rho-\theta)\int_{t}^{\infty}\bar{L}(s)ds-(\beta\rho-\theta)\int_{t+\frac{x}{\rho}}^{\infty}\bar{L}(s)ds}
\leq
\mathbb{E}_{X_{t}=x}\left[e^{-\theta(\tau_{t}-t)}\right]
\leq e^{-\beta x+(\beta\rho-\theta)\int_{t}^{\infty}\bar{L}(s)ds}.
\end{equation}

Notice that $\beta\rho-\theta=\lambda\int_{0}^{\infty}[1-e^{-\beta y}]p(y)dy\leq\lambda$
and thus
\begin{equation}
e^{(\beta\rho-\theta)\int_{\frac{x}{\rho}+t}^{\infty}\bar{L}(s)ds}
\leq e^{\lambda\int_{\frac{x}{\rho}+t}^{\infty}\bar{L}(s)ds}
\rightarrow 1,
\end{equation}
as $x\rightarrow\infty$, uniformly in $\theta>0$.

Therefore, as $x\rightarrow\infty$, 
\begin{equation}
\mathbb{E}_{X_{t}=x}\left[e^{-\theta\tau_{t}}\right]
=e^{-\beta x}e^{(\beta\rho-\theta)\int_{t}^{\infty}\bar{L}(s)ds-\theta t}(1+o(1)),
\end{equation}
where $o(1)$ is uniform in $\theta>0$.

It is well known in the actuarial literatures that 
the variable $\beta$, that is defined in \eqref{betaEqn}, 
has the following explicit but quite complicated expression in terms of $\theta$, see e.g. \cite{DeVylder}
and \cite{Lin}:
\begin{equation}\label{betaExpression}
\beta=\frac{\theta+\lambda}{\rho}-\sum_{n=1}^{\infty}\frac{(\lambda/\rho)^{n}}{n!}
\int_{0}^{\infty}y^{n-1}e^{-\frac{\lambda+\theta}{\rho}y}dP^{\ast n}(y),
\end{equation}
where $P^{\ast n}(y)=\mathbb{P}(Y_{1}+Y_{2}+\cdots+Y_{n}\leq y)$. 

Indeed, we can also get an expression for a function of $\beta$, 
say $F(\beta)$ where $F(\cdot)$ is an analytic function. 
By Lagrange's implicit function theorem, see e.g. Goulden and Jackson \cite{Goulden},
\footnote{See formula (1.5) in \cite{Lin}. But note that in (1.5) in \cite{Lin} there was a typo
and $h(\frac{\delta+\lambda}{c})c$ in (1.5) should be $h(\frac{\delta+\lambda}{c})$ in \cite{Lin}.}
for any analytic function $F(z)$:
\begin{equation}
F(\beta)=F\left(\frac{\theta+\lambda}{\rho}\right)
+\sum_{n=1}^{\infty}(-1)^{n}\frac{(\lambda/\rho)^{n}}{n!}\frac{d^{n-1}}{dz^{n-1}}
\left\{F'(z)\int_{0}^{\infty}e^{-zy}dP^{\ast n}(y)\right\}\bigg|_{z=\frac{\theta+\lambda}{\rho}}.
\end{equation}
In our case, we can take $F(\beta)=e^{-\beta(x-\rho\int_{t}^{\infty}\bar{L}(s)ds)}$ and we get
\begin{align}
&e^{-\beta(x-\rho\int_{t}^{\infty}\bar{L}(s)ds)}
\\
&=e^{-\frac{\theta+\lambda}{\rho}(x-\rho\int_{t}^{\infty}\bar{L}(s)ds)}
\nonumber
\\
\nonumber
&\qquad
+\sum_{n=1}^{\infty}(-1)^{n}\frac{(\lambda/\rho)^{n}}{n!}\frac{d^{n-1}}{dz^{n-1}}
\bigg\{-\left(x-\rho\int_{t}^{\infty}\bar{L}(s)ds\right)
\nonumber
\\
&\qquad\qquad\qquad\cdot
\int_{0}^{\infty}
e^{-z(x-\rho\int_{t}^{\infty}\bar{L}(s)ds+y)}dP^{\ast n}(y)\bigg\}\bigg|_{z=\frac{\theta+\lambda}{\rho}}
\nonumber
\\
&=e^{-\frac{\theta+\lambda}{\rho}(x-\rho\int_{t}^{\infty}\bar{L}(s)ds)}
\nonumber
\\
&\qquad
+\left(x-\rho\int_{t}^{\infty}\bar{L}(s)ds\right)\sum_{n=1}^{\infty}\frac{(\lambda/\rho)^{n}}{n!}
\int_{0}^{\infty}\left(x-\rho\int_{t}^{\infty}\bar{L}(s)ds+y\right)^{n-1}
\nonumber
\\
&\qquad\qquad\qquad\cdot
e^{-\frac{\theta+\lambda}{\rho}(x-\rho\int_{t}^{\infty}\bar{L}(s)ds+y)}dP^{\ast n}(y).
\nonumber
\end{align}
Applying the inverse Laplace transform
and use the fact that for any $c>0$, $\mathcal{L}^{-1}\{e^{-c\theta}\}=\delta\left(t-c\right)$
and, $\mathcal{L}\{f(t-c)1_{t\geq c}\}=e^{-c\theta}\mathcal{L}\{f(t)\}$, we get that
the probability density function $f(T;t,x)$ of the ruin time $\tau_{t}$
for at time $T$ is given by
\begin{align}\label{Similar}
f(T;t,x)
&=\mathcal{L}^{-1}\left\{e^{-\beta(x-\rho\int_{t}^{\infty}\bar{L}(s)ds)}e^{-\theta\int_{t}^{\infty}\bar{L}(s)ds-\theta t}\right\}
\\
&=e^{-\frac{\lambda x}{\rho}+\lambda\int_{t}^{\infty}\bar{L}(s)ds}\mathcal{L}^{-1}\left\{e^{-\frac{\theta}{\rho}x-\theta t}\right\}
\nonumber
\\
&\qquad
+\left(x-\rho\int_{t}^{\infty}\bar{L}(s)ds\right)\sum_{n=1}^{\infty}\frac{(\lambda/\rho)^{n}}{n!}
\int_{0}^{\infty}\left(x-\rho\int_{t}^{\infty}\bar{L}(s)ds+y\right)^{n-1}
\nonumber
\\
&\qquad\qquad\qquad\cdot
\mathcal{L}^{-1}\left\{e^{-\frac{\theta}{\rho}(x+y+\rho t)-\frac{\lambda}{\rho}(x+y)+\lambda\int_{t}^{\infty}\bar{L}(s)ds}\right\}dP^{\ast n}(y)
\nonumber
\\
&=e^{-\frac{\lambda x}{\rho}+\lambda\int_{t}^{\infty}\bar{L}(s)ds}
\delta\left(T-\frac{x}{\rho}-t\right)
\nonumber
\\
&\qquad
+\left(x-\rho\int_{t}^{\infty}\bar{L}(s)ds\right)\sum_{n=1}^{\infty}\frac{(\lambda/\rho)^{n}}{n!}
\int_{0}^{\infty}\left(x-\rho\int_{t}^{\infty}\bar{L}(s)ds+y\right)^{n-1}
\nonumber
\\
&\qquad\qquad\qquad\cdot
e^{-\frac{\lambda}{\rho}(x+y)+\lambda\int_{t}^{\infty}\bar{L}(s)ds}
\delta\left(T-\frac{x+y}{\rho}-t\right)dP^{\ast n}(y)
\nonumber
\\
&=e^{-\frac{\lambda x}{\rho}+\lambda\int_{t}^{\infty}\bar{L}(s)ds}
\delta\left(T-\frac{x}{\rho}-t\right)
\nonumber
\\
&\qquad
+\left(x-\rho\int_{t}^{\infty}\bar{L}(s)ds\right)\sum_{n=1}^{\infty}\frac{(\lambda/\rho)^{n}}{n!}
\left(\rho(T-t)-\rho\int_{t}^{\infty}\bar{L}(s)ds\right)^{n-1}
\nonumber
\\
&\qquad\qquad\qquad\cdot
e^{-\lambda(T-t)+\lambda\int_{t}^{\infty}\bar{L}(s)ds}
p^{\ast n}(\rho(T-t)-x)\cdot 1_{T>t+\frac{x}{\rho}},
\nonumber
\end{align}
where $\delta(\cdot)$ is the Dirac delta function.
\end{proof}

Next, let us give some examples
for which the asymptotics in Theorem \ref{ThmII} becomes explicit.
We first consider the asymptotics of Laplace transform in Theorem \ref{ThmII}.

\begin{remark}
Let us assume that $p(y)=\nu e^{-\nu y}$ for some $\gamma>0$.
Then from \eqref{alphaEqn}, we get 
$\beta=\frac{-(\rho\nu-\lambda-\theta)+\sqrt{(\rho\nu-\lambda-\theta)^{2}+4\rho\nu\theta}}{2\rho}$.

(i) (sub-exponential delay) Assume that $\bar{L}(t)=\frac{1}{(1+t)^{\gamma}}$ where $\gamma>1$. Then
\begin{equation}
\mathbb{E}_{X_{t}=x}[e^{-\theta\tau_{t}}]
\sim e^{-\frac{-(\rho\nu-\lambda-\theta)+\sqrt{(\rho\nu-\lambda-\theta)^{2}+4\rho\nu\theta}}{2\rho}x
+\frac{-(\rho\nu-\lambda+\theta)+\sqrt{(\rho\nu-\lambda-\theta)^{2}+4\rho\nu\theta}}{2}\frac{1}{\gamma-1}\frac{1}{(1+t)^{\gamma-1}}},
\end{equation}
as $x\rightarrow\infty$.

(ii) (exponential delay) Assume that $\bar{L}(t)=e^{-\gamma t}$ where $\gamma>0$. Then
\begin{equation}
\mathbb{E}_{X_{t}=x}[e^{-\theta\tau_{t}}]
\sim e^{-\frac{-(\rho\nu-\lambda-\theta)+\sqrt{(\rho\nu-\lambda-\theta)^{2}+4\rho\nu\theta}}{2\rho}x
+\frac{-(\rho\nu-\lambda+\theta)+\sqrt{(\rho\nu-\lambda-\theta)^{2}+4\rho\nu\theta}}{2}\frac{1}{\gamma}e^{-\gamma t}},
\end{equation}
as $x\rightarrow\infty$.

(iii) (super-exponential delay) Assume that $\bar{L}(t)=e^{-\gamma t^{2}}$ where $\gamma>0$. Then
\begin{equation}
\mathbb{E}_{X_{t}=x}[e^{-\theta\tau_{t}}]
\sim e^{-\frac{-(\rho\nu-\lambda-\theta)+\sqrt{(\rho\nu-\lambda-\theta)^{2}+4\rho\nu\theta}}{2\rho}x
+\frac{-(\rho\nu-\lambda+\theta)+\sqrt{(\rho\nu-\lambda-\theta)^{2}+4\rho\nu\theta}}{2}\sqrt{\frac{\pi}{\gamma}}(1-N(\sqrt{2\gamma}t))},
\end{equation} 
as $x\rightarrow\infty$,
where $N(\cdot)$ is the cumulative distribution function of a standard normal random variable
with mean $0$ and variance $1$.
\end{remark}

We next consider the asymptotics of the probability density function in Theorem \ref{ThmII}.

\begin{remark}
Let us assume that $p(y)=\nu e^{-\nu y}$ for some $\gamma>0$.
Then we get $p^{\ast n}(y)=\frac{\nu^{n}y^{n-1}}{(n-1)!}e^{-\nu y}$.

(i) (sub-exponential delay) Assume that $\bar{L}(t)=\frac{1}{(1+t)^{\gamma}}$ where $\gamma>1$. Then
\begin{align}
f(T;t,x)&\sim
\left(x-\frac{\rho}{\gamma-1}\frac{1}{(1+t)^{\gamma-1}}\right)\sum_{n=1}^{\infty}\frac{(\lambda/\rho)^{n}}{n!}
\left(\rho(T-t)-\frac{\rho}{\gamma-1}\frac{1}{(1+t)^{\gamma-1}}\right)^{n-1}
\\
&\qquad\qquad\qquad\cdot
e^{-\lambda(T-t)+\frac{\lambda}{\gamma-1}\frac{1}{(1+t)^{\gamma-1}}}
\frac{\nu^{n}(\rho(T-t)-x)^{n-1}}{(n-1)!}e^{-\nu(\rho(T-t)-x)},
\nonumber
\end{align}
for $T>\frac{x}{\rho}+t$ and $x\rightarrow\infty$.

(ii) (exponential delay) Assume that $\bar{L}(t)=e^{-\gamma t}$ where $\gamma>0$. Then
\begin{align}
f(T;t,x)&\sim
\left(x-\frac{\rho}{\gamma}e^{-\gamma t}\right)\sum_{n=1}^{\infty}\frac{(\lambda/\rho)^{n}}{n!}
\left(\rho(T-t)-\frac{\rho}{\gamma}e^{-\gamma t}\right)^{n-1}
\\
&\qquad\qquad\qquad\cdot
e^{-\lambda(T-t)+\frac{\lambda}{\gamma}e^{-\gamma t}}
\frac{\nu^{n}(\rho(T-t)-x)^{n-1}}{(n-1)!}e^{-\nu(\rho(T-t)-x)},
\nonumber
\end{align}
for $T>\frac{x}{\rho}+t$ and $x\rightarrow\infty$.

(iii) (super-exponential delay) Assume that $\bar{L}(t)=e^{-\gamma t^{2}}$ where $\gamma>0$. Then
\begin{align}
f(T;t,x)&\sim
\left(x-\rho\sqrt{\frac{\pi}{\gamma}}(1-N(\sqrt{2\gamma}t))\right)\sum_{n=1}^{\infty}\frac{(\lambda/\rho)^{n}}{n!}
\left(\rho(T-t)-\rho\sqrt{\frac{\pi}{\gamma}}(1-N(\sqrt{2\gamma}t))\right)^{n-1}
\\
&\qquad\qquad\qquad\cdot
e^{-\lambda(T-t)+\lambda\sqrt{\frac{\pi}{\gamma}}(1-N(\sqrt{2\gamma}t))}
\frac{\nu^{n}(\rho(T-t)-x)^{n-1}}{(n-1)!}e^{-\nu(\rho(T-t)-x)},
\nonumber
\end{align}
for $T>\frac{x}{\rho}+t$ and $x\rightarrow\infty$,
where $N(\cdot)$ is the cumulative distribution function of a standard normal random variable
with mean $0$ and variance $1$.
\end{remark}


$\tau_{t}$ can take the value $\infty$ when the net condition holds
and thus $\mathbb{E}_{X_{t}=x}[\tau_{t}]$ is infinity. On the contrary, 
when we assume $\lambda\mathbb{E}[Y_{1}]<\rho$, that is, the opposite
of the net condition holds, $\mathbb{E}_{X_{t}=x}[\tau_{t}]$ is finite and we 
are interested to estimate:
\begin{equation}
\mathbb{E}_{X_{t}=x}[\tau_{t}],
\qquad
\text{as $x\rightarrow\infty$},
\end{equation}

Assume $\lambda\mathbb{E}[Y_{1}]<\rho$ and let
\begin{equation}
v(x,t):=\mathbb{E}_{X_{t}=x}[\tau_{t}-t].
\end{equation}
Then, $v(x,t)$ satisfies the equation:
\begin{equation}
\frac{\partial v}{\partial t}
-\rho\frac{\partial v}{\partial x}
+\lambda L(t)\int_{0}^{\infty}[v(x+y,t)-v(x,t)]p(y)dy
+1=0,
\end{equation}
with the boundary condition $v(0,t)\equiv 0$.

We have the following asymptotic results for the large initial surplus.

\begin{theorem}\label{ThmIII}
Let us assume that $\lambda\mathbb{E}[Y_{1}]<\rho$. Then,
\begin{align}
&t+\frac{x}{\rho-\lambda\mathbb{E}[Y_{1}]}
-\frac{\lambda\mathbb{E}[Y_{1}]}{\rho-\lambda\mathbb{E}[Y_{1}]}\int_{t}^{\infty}\bar{L}(s)ds
\\
&\qquad
\leq
\mathbb{E}_{X_{t}=x}[\tau_{t}]\leq
t+\frac{x}{\rho-\lambda\mathbb{E}[Y_{1}]}
-\frac{\lambda\mathbb{E}[Y_{1}]}{\rho-\lambda\mathbb{E}[Y_{1}]}\int_{t}^{\infty}\bar{L}(s)ds
+\frac{\lambda\mathbb{E}[Y_{1}]}{\rho-\lambda\mathbb{E}[Y_{1}]}\int_{t+\frac{x}{\rho}}^{\infty}\bar{L}(s)ds.
\nonumber
\end{align}
Hence, as $x\rightarrow\infty$, we have
\begin{equation}
\mathbb{E}_{X_{t}=x}[\tau_{t}]\sim t+\frac{x}{\rho-\lambda\mathbb{E}[Y_{1}]}
-\frac{\lambda\mathbb{E}[Y_{1}]}{\rho-\lambda\mathbb{E}[Y_{1}]}\int_{t}^{\infty}\bar{L}(s)ds.
\end{equation}
\end{theorem}

\begin{proof}
Let us define:
\begin{equation}
w(x,t):=\frac{x}{\rho-\lambda\mathbb{E}[Y_{1}]}
-\frac{\lambda\mathbb{E}[Y_{1}]}{\rho-\lambda\mathbb{E}[Y_{1}]}\int_{t}^{\infty}\bar{L}(s)ds.
\end{equation}
Then, we can compute that
\begin{align}
&\frac{\partial w}{\partial t}
-\rho\frac{\partial w}{\partial x}
+\lambda L(t)\int_{0}^{\infty}[w(x+y,t)-w(x,t)]p(y)dy+1
\\
&=\frac{\lambda\mathbb{E}[Y_{1}]}{\rho-\lambda\mathbb{E}[Y_{1}]}\bar{L}(t)-\bar{L}(t)\frac{\lambda\mathbb{E}[Y_{1}]}{\rho-\lambda\mathbb{E}[Y_{1}]}=0.
\nonumber
\end{align}
Let $\tau_{t}^{K}$ be the first time after $t$ the process $X_{t}$ hits above $K>x$. 
Therefore, by It\^{o}'s formula and optional stopping theorem, we have
\begin{align}
\mathbb{E}_{X_{t}=x}\left[w\left(X_{\tau_{t}\wedge\tau_{t}^{K}},\tau_{t}\wedge\tau_{t}^{K}\right)\right]
&=w(x,t)-\mathbb{E}_{X_{t}=x}[\tau_{t}\wedge\tau_{t}^{K}-t]
\\
&=\frac{x}{\rho-\lambda\mathbb{E}[Y_{1}]}
-\frac{\lambda\mathbb{E}[Y_{1}]}{\rho-\lambda\mathbb{E}[Y_{1}]}\int_{t}^{\infty}\bar{L}(s)ds
-\mathbb{E}_{X_{t}=x}[\tau_{t}\wedge\tau_{t}^{K}-t].
\nonumber
\end{align}
On the other hand,
\begin{equation}
\mathbb{E}_{X_{t}=x}\left[w\left(X_{\tau_{t}\wedge\tau_{t}^{K}},\tau_{t}\wedge\tau_{t}^{K}\right)\right]
=
\frac{\mathbb{E}\left[X_{\tau_{t}\wedge\tau_{t}^{K}}\right]}{\rho-\lambda\mathbb{E}[Y_{1}]}
-\frac{\lambda\mathbb{E}[Y_{1}]}{\rho-\lambda\mathbb{E}[Y_{1}]}
\mathbb{E}_{X_{t}=x}\left[\int_{\tau_{t}\wedge\tau_{t}^{K}}^{\infty}\bar{L}(s)ds\right].
\end{equation}
It is easy to see that
\begin{equation}\label{subI}
\mathbb{E}_{X_{t}=x}\left[X_{\tau_{t}\wedge\tau_{t}^{K}}\right]
=\mathbb{E}_{X_{t}=x}\left[X_{\tau_{t}^{K}}|\tau_{t}^{K}<\tau_{t}\right]\mathbb{P}(\tau_{t}^{K}<\tau_{t}|X_{t}=x).
\end{equation}
On the other hand, there exists $\alpha>0$ such that
\begin{equation}
-\rho\alpha+\lambda\int_{0}^{\infty}[e^{\alpha y}-1]p(y)dy=0,
\end{equation}
and therefore
\begin{equation}
\mathbb{E}_{X_{t}=x}\left[e^{\alpha X_{\tau_{t}\wedge\tau_{t}^{K}}}\right]
=e^{\alpha x}=1-\mathbb{P}(\tau_{t}^{K}<\tau_{t})+\mathbb{E}_{X_{t}=x}\left[e^{\alpha X_{\tau_{t}^{K}}}|\tau_{t}^{K}<\tau_{t}\right]
\mathbb{P}(\tau_{t}^{K}<\tau_{t}|X_{t}=x),
\end{equation}
which implies that
\begin{equation}\label{subII}
\mathbb{P}(\tau_{t}^{K}<\tau_{t}|X_{t}=x)=\frac{e^{\alpha x}-1}{\mathbb{E}_{X_{t}=x}\left[e^{\alpha X_{\tau_{t}^{K}}}|\tau_{t}^{K}<\tau_{t}\right]-1}.
\end{equation}
Substituting \eqref{subII} into \eqref{subI}, we get
\begin{equation}
\mathbb{E}_{X_{t}=x}\left[X_{\tau_{t}\wedge\tau_{t}^{K}}\right]
=\mathbb{E}_{X_{t}=x}\left[X_{\tau_{t}^{K}}|\tau_{t}^{K}<\tau_{t}\right]
\frac{e^{\alpha x}-1}{\mathbb{E}_{X_{t}=x}\left[e^{\alpha X_{\tau_{t}^{K}}}|\tau_{t}^{K}<\tau_{t}\right]-1}.
\end{equation}
Since $\alpha>0$ and $X_{\tau_{t}^{K}}\geq K$, we have
\begin{equation}
\mathbb{E}_{X_{t}=x}\left[X_{\tau_{t}\wedge\tau_{t}^{K}}\right]
\rightarrow 0,
\qquad\text{as $K\rightarrow\infty$}.
\end{equation}
Finally by monotone convergence theorem,
\begin{equation}
\lim_{K\rightarrow\infty}\mathbb{E}_{X_{t}=x}[\tau_{t}\wedge\tau_{t}^{K}-t]
=\mathbb{E}_{X_{t}=x}[\tau_{t}-t],
\qquad
\lim_{K\rightarrow\infty}\mathbb{E}_{X_{t}=x}\left[\int_{\tau_{t}\wedge\tau_{t}^{K}}^{\infty}\bar{L}(s)ds\right]
=\mathbb{E}_{X_{t}=x}\left[\int_{\tau_{t}}^{\infty}\bar{L}(s)ds\right].
\end{equation}
Therefore,
\begin{equation}
\mathbb{E}_{X_{t}=x}[\tau_{t}]=t+
\frac{x}{\rho-\lambda\mathbb{E}[Y_{1}]}
-\frac{\lambda\mathbb{E}[Y_{1}]}{\rho-\lambda\mathbb{E}[Y_{1}]}\int_{t}^{\infty}\bar{L}(s)ds
+\frac{\lambda\mathbb{E}[Y_{1}]}{\rho-\lambda\mathbb{E}[Y_{1}]}\mathbb{E}_{X_{t}=x}\left[\int_{\tau_{t}}^{\infty}\bar{L}(s)ds\right].
\end{equation}
Finally, notice that $\tau_{t}\geq t+\frac{x}{\rho}$. Therefore, the proof is complete.
\end{proof}

Next, let us give some examples
for which the asymptotics in Theorem \ref{ThmIII} becomes explicit.

\begin{remark}
Let us assume that $p(y)=\nu e^{-\nu y}$ for some $\gamma>0$ and $\frac{\lambda}{\nu}<\rho$. 
Then $\mathbb{E}[Y_{1}]=\frac{1}{\nu}$.

(i) (sub-exponential delay) Assume that $\bar{L}(t)=\frac{1}{(1+t)^{\gamma}}$ where $\gamma>1$. Then
\begin{equation}
\mathbb{E}_{X_{t}=x}[\tau_{t}]\sim t+\frac{x}{\rho-\frac{\lambda}{\nu}}
-\frac{\lambda}{\rho\nu-\lambda}\frac{1}{\gamma-1}\frac{1}{(1+t)^{\gamma-1}},
\end{equation}
as $x\rightarrow\infty$.

(ii) (exponential delay) Assume that $\bar{L}(t)=e^{-\gamma t}$ where $\gamma>0$. Then
\begin{equation}
\mathbb{E}_{X_{t}=x}[\tau_{t}]\sim t+\frac{x}{\rho-\frac{\lambda}{\nu}}
-\frac{\lambda}{\rho\nu-\lambda}\frac{1}{\gamma}e^{-\gamma t},
\end{equation}
as $x\rightarrow\infty$.

(iii) (super-exponential delay) Assume that $\bar{L}(t)=e^{-\gamma t^{2}}$ where $\gamma>0$. Then
\begin{equation}
\mathbb{E}_{X_{t}=x}[\tau_{t}]\sim t+\frac{x}{\rho-\frac{\lambda}{\nu}}
-\frac{\lambda}{\rho\nu-\lambda}\sqrt{\frac{\pi}{\gamma}}(1-N(\sqrt{2\gamma}t)),
\end{equation} 
as $x\rightarrow\infty$,
where $N(\cdot)$ is the cumulative distribution function of a standard normal random variable
with mean $0$ and variance $1$.
\end{remark}


\section{Exactly Solvable Models}

We have already studied the large initial surplus asymptotics
for ruin probabilities and ruin times. Despite the difficulty
to get closed-form solution for the equation \eqref{survivalEqn}, 
we will show in this section that for special cases, the ruin probability
and ruin time can be solved in closed-form.

When the delay time is constant, the model is exactly solvable. 
Let us assume that the delay time is $\ell>0$.  Then $\psi(x,t)=\mathbb{P}(\tau_{t}<\infty|X_{t}=x)$
satisfies the equation:
\begin{equation}\label{PDEI}
\frac{\partial\psi}{\partial t}
-\rho\frac{\partial\psi}{\partial x}=0,
\qquad
\text{for any $t<\ell$},
\end{equation}
and
\begin{equation}
\frac{\partial\psi}{\partial t}
-\rho\frac{\partial\psi}{\partial x}
+\lambda\int_{0}^{\infty}[\psi(x+y,t)-\psi(x)]p(y)dy=0,
\qquad
\text{for any $t\geq\ell$},
\end{equation}
with the boundary condition that $\psi(0,t)\equiv 1$.

\begin{theorem}\label{ConstThm}
Assume that delay time is $\ell$, a positive constant. Then,
\begin{equation}
\mathbb{P}(\tau_{t}<\infty|X_{t}=x)
=
\begin{cases}
1 &\text{for $t<\ell$ and $x\leq\rho(\ell-t)$}
\\
e^{-\alpha(x-\rho(\ell-t))} &\text{for $t<\ell$ and $x>\rho(\ell-t)$}
\\
e^{-\alpha x} &\text{for $t\geq\ell$}
\end{cases},
\end{equation}
where $\alpha$ is the unique positive value that satisfies:
\begin{equation}
\alpha\rho+\lambda\int_{0}^{\infty}[e^{-\alpha y}-1]p(y)dy=0.
\end{equation}
\end{theorem}

\begin{proof}
It is easy to see that
\begin{equation}
\psi(x,t)=e^{-\alpha x},
\qquad
\text{for $t\geq\ell$},
\end{equation}
where $\alpha$ is the unique positive value that satisfies:
\begin{equation}
\alpha\rho+\lambda\int_{0}^{\infty}[e^{-\alpha y}-1]p(y)dy=0.
\end{equation}
For $t<\ell$, from the PDE \eqref{PDEI}, by characteristic method, $\psi(x,t)=F(\rho t+x)$
for some function $F$. Indeed, the surplus process decreases
at the constant rate $\rho$ before time $\ell$ and thus
\begin{equation}
\psi(x,t)=1,
\qquad\text{for $t<\ell$ and $x\leq\rho(\ell-t)$},
\end{equation}
and otherwise right before time $\ell$ the surplus is $x-\rho(\ell-t)$
and then the ruin probability is the same as the infinite horizon ruin probability with constant Poisson rate
$\lambda$ with initial surplus $x-\rho(\ell-t)$ and hence
\begin{equation}
\psi(x,t)=e^{-\alpha(x-\rho(\ell-t))},
\qquad\text{for $t<\ell$ and $x>\rho(\ell-t)$}.
\end{equation}
\end{proof}

Following the same argument as in the proof of Theorem \ref{ConstThm}, 
it is easy to see that for any $\theta>0$,
\begin{equation}
\mathbb{E}[e^{-\theta\tau_{t}}|X_{t}=x]
=
\begin{cases}
e^{-\theta t-\theta\frac{x}{\rho}} &\text{for $t<\ell$ and $x\leq\rho(\ell-t)$}
\\
e^{-\theta t}e^{-\beta(x-\rho(\ell-t))} &\text{for $t<\ell$ and $x>\rho(\ell-t)$}
\\
e^{-\theta t}e^{-\beta x} &\text{for $t\geq\ell$}
\end{cases},
\end{equation}
where $\beta$ is the unique positive value that satisfies:
\begin{equation}
\beta\rho+\lambda\int_{0}^{\infty}[e^{-\beta y}-1]p(y)dy-\theta=0.
\end{equation}
Then, we can use the Laplace transform of the ruin
time to obtain the probability density function
of the ruin time.

\begin{theorem}
Assume that delay time is $\ell$, a positive constant. Then,
the probability density function of the ruin time is given by
\begin{equation}
f(T;t,x)
=\delta\left(T-t-\frac{x}{\rho}\right),
\qquad\text{for $t<\ell$ and $x\leq\rho(\ell-t)$},
\end{equation}
and for $t<\ell$ and $x>\rho(\ell-t)$, 
\begin{align}
f(T;t,x)
&=e^{-\frac{\lambda}{\rho}(x-\rho(\ell-t))}\delta\left(T-t-\frac{(x-\rho(\ell-t))}{\rho}\right)
\\
&\qquad
+(x-\rho(\ell-t))\sum_{n=1}^{\infty}\frac{(\lambda/\rho)^{n}}{n!}
(\rho(T-t))^{n-1}e^{-\lambda(T-t)}
\nonumber
\\
&\qquad\qquad\qquad
\cdot p^{\ast n}(\rho(T-t)-(x-\rho(\ell-t)))\cdot 1_{T\geq t+\frac{(x-\rho(\ell-t))}{\rho}},
\nonumber
\end{align}
and for $t\geq\ell$,
\begin{align}
f(T;t,x)
&=e^{-\frac{\lambda}{\rho}x}\delta\left(T-t-\frac{x}{\rho}\right)
\\
&\qquad
+x\sum_{n=1}^{\infty}\frac{(\lambda/\rho)^{n}}{n!}
(\rho(T-t))^{n-1}e^{-\lambda(T-t)}p^{\ast n}(\rho(T-t)-x)\cdot 1_{T\geq t+\frac{x}{\rho}},
\nonumber
\end{align}
where $p^{\ast n}(t)$ is the probability density function of $Y_{1}+Y_{2}+\cdots+Y_{n}$.
\end{theorem}

\begin{proof}
First of all, for $t<\ell$ and $x\leq\rho(\ell-t)$, it is clear that
\begin{equation}
\tau_{t}=t+\frac{x}{\rho},
\qquad
\text{a.s.}
\end{equation}
Notice that similar as in \eqref{Similar}, we can compute that
\begin{equation}
\mathcal{L}^{-1}\{e^{-\beta x}\}
=e^{-\frac{\lambda}{\rho}x}\delta\left(t-\frac{x}{\rho}\right)
+x\sum_{n=1}^{\infty}\frac{(\lambda/\rho)^{n}}{n!}
(\rho t)^{n-1}e^{-\lambda t}p^{\ast n}(\rho t-x)\cdot 1_{t\geq\frac{x}{\rho}},
\end{equation}
where $p^{\ast n}(t)$ is the probability density function of $Y_{1}+Y_{2}+\cdots+Y_{n}$.
Therefore, for any $t\geq\ell$, the probability density function of $\tau_{t}$ is 
given by
\begin{align}
f(T;t,x)
=e^{-\frac{\lambda}{\rho}x}\delta\left(T-t-\frac{x}{\rho}\right)
+x\sum_{n=1}^{\infty}\frac{(\lambda/\rho)^{n}}{n!}
(\rho(T-t))^{n-1}e^{-\lambda(T-t)}p^{\ast n}(\rho(T-t)-x)\cdot 1_{T\geq t+\frac{x}{\rho}},
\end{align}
and for $t<\ell$ and $x>\rho(\ell-t)$, 
\begin{align}
f(T;t,x)
&=e^{-\frac{\lambda}{\rho}(x-\rho(\ell-t))}\delta\left(T-t-\frac{(x-\rho(\ell-t))}{\rho}\right)
\\
&\qquad
+(x-\rho(\ell-t))\sum_{n=1}^{\infty}\frac{(\lambda/\rho)^{n}}{n!}
(\rho(T-t))^{n-1}e^{-\lambda(T-t)}
\nonumber
\\
&\qquad\qquad\qquad
\cdot p^{\ast n}(\rho(T-t)-(x-\rho(\ell-t)))\cdot 1_{T\geq t+\frac{(x-\rho(\ell-t))}{\rho}}.
\nonumber
\end{align}
\end{proof}


More generally, assume that the delay length is shorter than $\ell>0$.
Then, $L(t)=1$ for any $t\geq\ell$ and $L(t)<1$ for any $t<\ell$.
For any $t\geq\ell$, the ruin probability $\psi(x,t)=\mathbb{P}(\tau_{t}<\infty|X_{t}=x)$
satisfies the equation
\begin{equation}
\frac{\partial\psi}{\partial t}
-\rho\frac{\partial\psi}{\partial x}
+\lambda\int_{0}^{\infty}[\psi(x+y,t)-\psi(x)]p(y)dy=0,
\end{equation}
which implies that
\begin{equation}
\psi(x,t)=e^{-\alpha x},
\qquad
\text{for any $t\geq\ell$}.
\end{equation}
For any $t<\ell$, we have the following partial result. 

\begin{theorem}\label{CompactThm}
Assume $L(t)=1$ for any $t\geq\ell>0$.
For $t<\ell$ and $x>\rho(\ell-t)$,
\begin{equation}
\mathbb{P}(\tau_{t}<\infty|X_{t}=x)
=e^{-\alpha x+\alpha\rho\int_{t}^{\ell}\bar{L}(s)ds}.
\end{equation}
\end{theorem}

\begin{proof}
When $t<\ell$ and $x>\rho(\ell-t)$, the ruin cannot occur
during the time interval $[t,\ell]$ and conditional
on $X_{\ell}=x$, the ruin probability is $e^{-\alpha x}$.
Therefore for $t<\ell$ and $x>\rho(\ell-t)$,
\begin{equation}
\psi(x,t)=\mathbb{E}_{X_{t}=x}\left[e^{-\alpha X_{\ell}}\right].
\end{equation}
By Dynkin's formula, we can compute that
\begin{equation}
\mathbb{E}_{X_{t}=x}\left[e^{-\alpha X_{\ell}}\right]=e^{-\alpha x}
+\int_{t}^{\ell}\left(\rho\alpha+\lambda L(s)\int_{0}^{\infty}[e^{-\alpha y}-1]p(y)dy\right)\mathbb{E}_{X_{t}=x}[e^{-\alpha X_{s}}]ds,
\end{equation}
which implies that
\begin{equation}
\mathbb{E}_{X_{t}=x}\left[e^{-\alpha X_{\ell}}\right]
=e^{-\alpha x}e^{\int_{t}^{\ell}\left(\rho\alpha+\lambda L(s)\int_{0}^{\infty}[e^{-\alpha y}-1]p(y)dy\right)ds}
=e^{-\alpha x+\alpha\rho\int_{t}^{\ell}\bar{L}(s)ds}.
\end{equation}
\end{proof}

\begin{remark}
The result for Theorem \ref{CompactThm} holds when the distribution
of the delay time has compact support. Indeed, the techniques used in
the proof of Theorem \ref{CompactThm} can also be used 
to provide an alternative proof of Theorem \ref{ThmI}.
For any $\epsilon>0$, there exists $\ell>0$ such that
$L(\ell)\geq 1-\epsilon$. Then,
\begin{equation}\label{IneqI}
e^{-\alpha x}\leq\mathbb{P}(\tau_{\ell}<\infty|X_{\ell}=x)\leq e^{-\alpha_{\epsilon}x},
\end{equation}
where $\alpha_{\epsilon}$ is the unique positive value that satisfies:
\begin{equation}
\rho\alpha_{\epsilon}+\lambda(1-\epsilon)\int_{0}^{\infty}[e^{-\alpha_{\epsilon}y}-1]p(y)dy=0.
\end{equation}
For sufficiently large $x$, i.e. $x>\rho(\ell-t)$, following
the same arguments as in the proof of Theorem \ref{CompactThm} and also by \eqref{IneqI}, we have
\begin{equation}
e^{-\alpha x+\alpha\rho\int_{t}^{\ell}\bar{L}(s)ds}
\leq\mathbb{P}(\tau_{t}<\infty|X_{t}=x)
\leq e^{-\alpha_{\epsilon} x+\frac{\rho\epsilon\alpha_{\epsilon}}{1-\epsilon}
+\frac{\rho\alpha_{\epsilon}}{1-\epsilon}\int_{t}^{\ell}\bar{L}(s)ds}.
\end{equation}
\end{remark}

\begin{remark}
An interesting open problem is to compute the ruin probability and the probability
density function of the ruin time for the other regions in Theorem \ref{CompactThm}, that is,
for $x<\rho(\ell-t)$ and $t<\ell$. Even if there is no closed-form in general, it would
still be interesting to see if closed-form formulas exist for some special examples of $L(t)$ and $p(y)$.
\end{remark}

Assume $L(t)=1$ for any $t\geq\ell>0$. 
For $t<\ell$ and $x>\rho(\ell-t)$, we can obtain
the probability density function of the ruin time $\tau_{t}$ in closed-form.

\begin{theorem}
Assume $L(t)=1$ for any $t\geq\ell>0$. 
For $t<\ell$ and $x>\rho(\ell-t)$, we can obtain
the probability density function of the ruin time $\tau_{t}$:
\begin{align}
f(T;t,x)&=e^{-\frac{\lambda}{\rho}(x-\rho\int_{t}^{\ell}\bar{L}(s)ds)}
\delta\left(T-t-\frac{x}{\rho}\right)
\\
&\qquad
+\left(x-\rho\int_{t}^{\ell}\bar{L}(s)ds\right)\sum_{n=1}^{\infty}\frac{(\lambda/\rho)^{n}}{n!}
\left(\rho\left(T-\ell+\int_{t}^{\ell}L(s)ds\right)\right)^{n-1}
\nonumber
\\
&\qquad\qquad\cdot
e^{-\lambda (T-\ell+\int_{t}^{\ell}L(s)ds)}
p^{\ast n}\left(\rho(T-t)-x\right)
\cdot 1_{T\geq t+\frac{x}{\rho}},
\nonumber
\end{align}
where $p^{\ast n}(t)$ is the probability density function of $Y_{1}+Y_{2}+\cdots+Y_{n}$.
\end{theorem}

\begin{proof}
For any $\theta>0$, since when $t<\ell$ and $x>\rho(\ell-t)$
the ruin cannot occur during the time interval $[t,\ell]$, we have
\begin{align}
\mathbb{E}_{X_{t}=x}[e^{-\theta\tau_{t}}]
=e^{-\theta\ell}e^{-\beta x}e^{\int_{t}^{\ell}\left(\rho\beta+\lambda L(s)\int_{0}^{\infty}[e^{-\beta y}-1]p(y)dy\right)ds}
=e^{-\theta\ell}e^{-\beta x+\beta\rho\int_{t}^{\ell}\bar{L}(s)ds+\theta\int_{t}^{\ell}L(s)ds}.
\end{align}
Notice that similar as in \eqref{Similar}, we can compute that
\begin{equation}
\mathcal{L}^{-1}\{e^{-\beta x}\}
=e^{-\frac{\lambda}{\rho}x}\delta\left(t-\frac{x}{\rho}\right)
+x\sum_{n=1}^{\infty}\frac{(\lambda/\rho)^{n}}{n!}
(\rho t)^{n-1}e^{-\lambda t}p^{\ast n}(\rho t-x)\cdot 1_{t\geq\frac{x}{\rho}},
\end{equation}
where $p^{\ast n}(t)$ is the probability density function of $Y_{1}+Y_{2}+\cdots+Y_{n}$.
Therefore, the probability density function of the ruin time $\tau_{t}$
is given by
\begin{align}
f(T;t,x)
&=e^{-\frac{\lambda}{\rho}(x-\rho\int_{t}^{\ell}\bar{L}(s)ds)}
\delta\left(T-\ell+\int_{t}^{\ell}L(s)ds-\frac{(x-\rho\int_{t}^{\ell}\bar{L}(s)ds)}{\rho}\right)
\\
&\qquad
+\left(x-\rho\int_{t}^{\ell}\bar{L}(s)ds\right)\sum_{n=1}^{\infty}\frac{(\lambda/\rho)^{n}}{n!}
\left(\rho\left(T-\ell+\int_{t}^{\ell}L(s)ds\right)\right)^{n-1}
\nonumber
\\
&\qquad\qquad\cdot
e^{-\lambda (T-\ell+\int_{t}^{\ell}L(s)ds)}
p^{\ast n}\left(\rho \left(T-\ell+\int_{t}^{\ell}L(s)ds\right)-\left(x-\rho\int_{t}^{\ell}\bar{L}(s)ds\right)\right)
\nonumber
\\
&\qquad\qquad\qquad\qquad
\cdot 1_{T-\ell+\int_{t}^{\ell}L(s)ds\geq\frac{(x-\rho\int_{t}^{\ell}\bar{L}(s)ds)}{\rho}},
\nonumber
\\
&=e^{-\frac{\lambda}{\rho}(x-\rho\int_{t}^{\ell}\bar{L}(s)ds)}
\delta\left(T-t-\frac{x}{\rho}\right)
\nonumber
\\
&\qquad
+\left(x-\rho\int_{t}^{\ell}\bar{L}(s)ds\right)\sum_{n=1}^{\infty}\frac{(\lambda/\rho)^{n}}{n!}
\left(\rho\left(T-\ell+\int_{t}^{\ell}L(s)ds\right)\right)^{n-1}
\nonumber
\\
&\qquad\qquad\cdot
e^{-\lambda (T-\ell+\int_{t}^{\ell}L(s)ds)}
p^{\ast n}\left(\rho(T-t)-x\right)
\cdot 1_{T\geq t+\frac{x}{\rho}}.
\nonumber
\end{align}
The intuition behind $\tau_{t}\geq t+\frac{x}{\rho}$ is that the ruin can
not occur during the time interval $[t,\ell]$. Starting at time $\ell$, the surplus is $x-\rho(\ell-t)$
and thus it will take at least the amount of time $\frac{x-\rho(\ell-t)}{\rho}$ after time $\ell$ 
to get ruined. Thus, $\tau_{t}\geq\ell+\frac{x-\rho(\ell-t)}{\rho}=t+\frac{x}{\rho}$ a.s.
\end{proof}


\section{Numerical Studies}

Let us first illustrate the case when the delay time is constant $\ell$. 
Let us assume that $p(y)=\nu e^{-\nu y}$, then $\alpha=\frac{\lambda}{\rho}-\nu$, 
and by Theorem \ref{ConstThm}, we have
\begin{equation}
\psi(x,t):=\mathbb{P}(\tau_{t}<\infty|X_{t}=x)
=
\begin{cases}
1 &\text{for $t<\ell$ and $x\leq\rho(\ell-t)$}
\\
e^{-(\frac{\lambda}{\rho}-\nu)(x-\rho(\ell-t))} &\text{for $t<\ell$ and $x>\rho(\ell-t)$}
\\
e^{-(\frac{\lambda}{\rho}-\nu) x} &\text{for $t\geq\ell$}
\end{cases}.
\end{equation}
Take $\ell=2$, $\nu=0.01$, $\lambda=0.02$ and $\rho=1$. We illustrate
the ruin probability in a heatmap plot in Figure \ref{HeatMap}.
The statistics are also summarized in Table \ref{HeatMapTable}.

\begin{table}[htb]
\centering 
\caption{Illustration of the ruin probability against
the initial surplus and various present times. Here, we take $\ell=2$, $\rho=1$, $\nu=0.01$ and $\lambda=0.02$.}
\begin{tabular}{|c||c|c|c|c|c|c|} 
\hline 
$\psi(x,t)$ & $x=0.5$ & $x=1.5$ & $x=2.5$ & $x=3.5$ & $x=4.5$ 
\\
\hline
\hline
$t=0.5$ & 1.000 & 1.000 & 0.990 & 0.980 & 0.970
\\
\hline
$t=1.5$ & 1.000 & 0.990 & 0.980 & 0.970 & 0.961 
\\
\hline
$t=2.5$ & 0.995 & 0.985 & 0.975 & 0.966 & 0.956 
\\
\hline
$t=3.5$ & 0.995 & 0.985 & 0.975 & 0.966 & 0.956 
\\
\hline
$t=4.5$ & 0.995 & 0.985 & 0.975 & 0.966 & 0.956 
\\
\hline 
\end{tabular}
\label{HeatMapTable} 
\end{table}

\begin{figure}[htb]
\begin{center}
\includegraphics[scale=0.90]{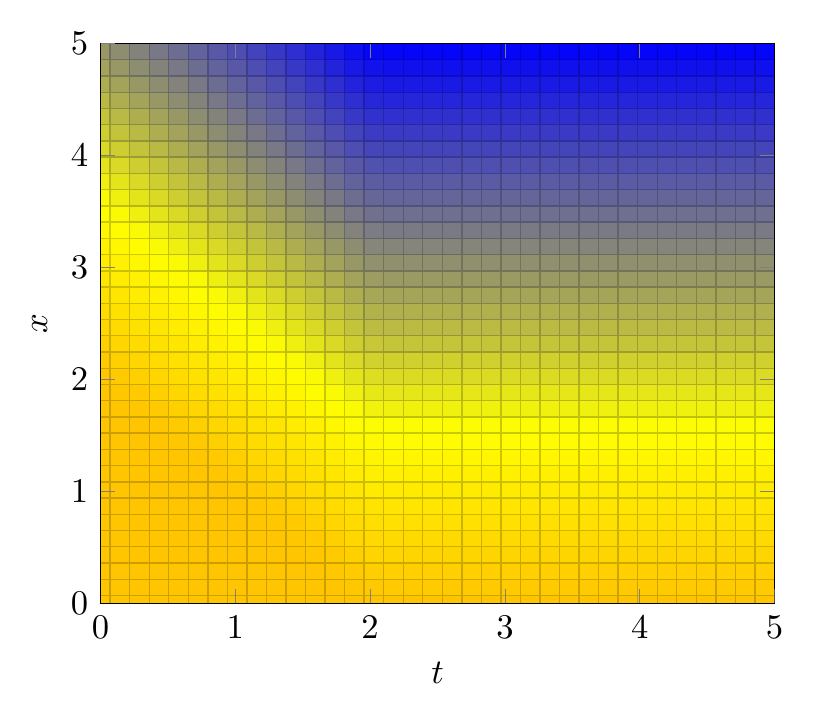}
\caption{This is the heatmap plot of the ruin probability as a function
of the initial surplus $x$ and the present time $t$. The ruin probability is constantly $1$
in the region $t<\ell$ and $x\leq\rho(\ell-t)$ and the ruin probability depends only
on $x$ in the region $t\geq\ell$. Here, we take $\ell=2$, $\nu=0.01$, $\lambda=0.02$ and $\rho=1$.}
\label{HeatMap}
\end{center}
\end{figure}

Next, let us consider the case when the delay time is bounded. 
As an illustration, let us consider $L(t)=t$ for $t<1$ and $L(t)=1$ otherwise.
Then, $\ell=1$.
That is, the delay time has a uniform distribution on $[0,1]$. Let us also
take $p(y)=\nu e^{-\nu y}$. Then, for any $t\leq\ell$ and $x>\rho(\ell-t)$, by Theorem \ref{CompactThm}, 
we have
\begin{equation}
\psi(x,t):=\mathbb{P}(\tau_{t}<\infty|X_{t}=x)=e^{-(\frac{\lambda}{\rho}-\nu)x+(\frac{\lambda}{\rho}-\nu)\frac{\rho}{2}(1-t)^{2}}.
\end{equation}
Let us take $\rho=\nu=1$, $\lambda=2$. We illustrate the ruin probability
by making a plot of the ruin probability $\psi(x):=\psi(x,t)$ against the initial surplus $x$
for fixed present times $t$ in Figure \ref{CompactPicture}.
The statistics are also summarized in Table \ref{CompactTable}.
Note that in Table \ref{CompactTable}, the entries of N.A. do not mean
that the ruin probabilities do not exist. They are labelled N.A. only because
Theorem \ref{CompactThm} does not provide closed-form formulas for the ruin probabilities
when $x<\rho(\ell-t)$ and $t<\ell$. It will be an interesting problem to obtain
closed-form formulas at least for some special cases in this region. 

\begin{table}[htb]
\centering 
\caption{Illustration of the ruin probability against
the initial surplus for various present times. Here, we take $\ell=1$, $\rho=\nu=1$ and $\lambda=2$.}
\begin{tabular}{|c||c|c|c|c|c|c|} 
\hline 
$\psi(x,t)$ & $x=0.2$ & $x=0.4$ & $x=0.6$ & $x=0.8$ & $x=1.0$ 
\\
\hline
\hline
$t=0.25$ & N.A. & N.A. & N.A. & 0.595 & 0.487 
\\
\hline
$t=0.50$ & N.A. & N.A. & 0.622 & 0.509 & 0.417 
\\
\hline
$t=0.75$ & N.A. & 0.692 & 0.566 & 0.464 & 0.380 
\\
\hline
$t=1.00$ & 0.819 & 0.670 & 0.549 & 0.449 & 0.368 
\\
\hline 
\end{tabular}
\label{CompactTable} 
\end{table}

\begin{figure}[htb]
\begin{center}
\includegraphics[scale=0.90]{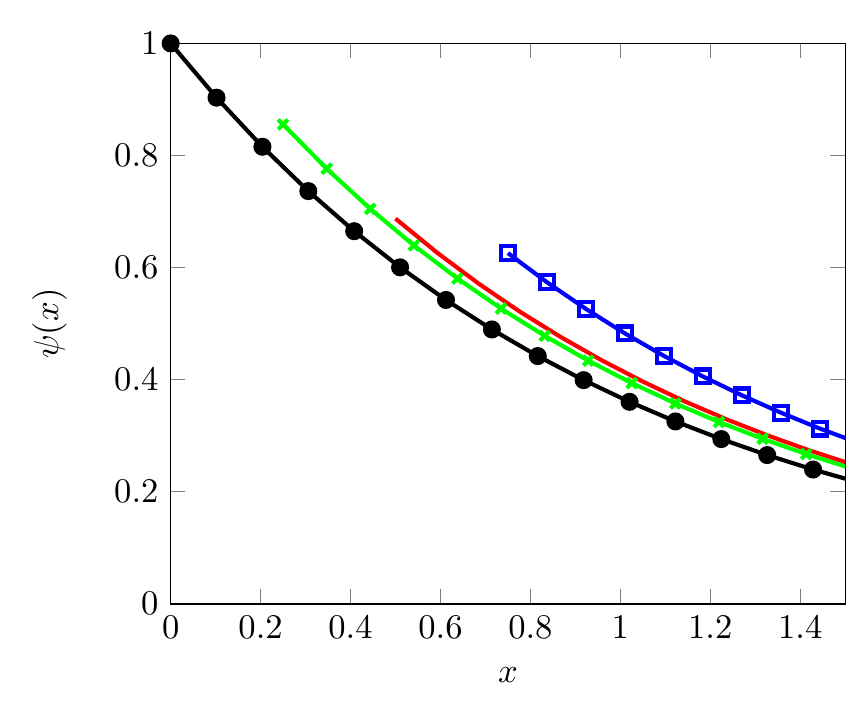}
\caption{This is a plot of the ruin probability $\psi(x)$ against the initial surplus $x$
for various present times $t=0.25$, $t=0.50$, $t=0.75$ and $t=1.0$. 
Here, we take $\ell=1$, $\rho=\nu=1$ and $\lambda=2$. The blue line marked with square
corresponds to $t=0.25$; the red solid line corresponds to $t=0.50$; the green line marked with x 
corresponds to $t=0.75$; the black line marked with circle corresponds to $t=1.0$.
We only make the plots for the region $t\leq\ell$ and $x>\rho(\ell-t)$ since
this is the region where we have analytical tractability according to Theorem \ref{CompactThm}.}
\label{CompactPicture}
\end{center}
\end{figure}


\section*{Acknowledgements}
The author thanks the anonymous referee and the editor for helpful suggestions
and comments that have greatly improved the quality of the manuscript. 
The author is also grateful to Arash Fahim and Hongbiao Zhao for helpful discussions.
The author acknowledges the support from NSF Grant DMS-1613164.



\begin{thebibliography}{99}

\bibitem{Afonso}
Afonso, L. B., Cardoso, R. M. R. and A. D. Eg\'{i}dio dos Reis. (2013).
Dividend problems in the dual risk model.
\textit{Insurance: Mathematics and Economics}.
\textbf{53}, 906-918.

\bibitem{AA}
Albrecher, H. and S. Asmussen. (2006).
Ruin probabilities and aggregate claims distributions for shot noise Cox processes.
\textit{Scandinavian Actuarial Journal}.
\textbf{2}, 86-110.

\bibitem{Albrecher}
Albrecher, H., Badescu, A. and D. Landriault. (2008).
On the dual risk model with tax payments.
\textit{Insurance: Mathematics and Economics}.
\textbf{42}, 1086-1094.

\bibitem{AvanziII}
Avanzi, B., Cheung, E. C. K., Wong, B. and J. K. Woo. (2013).
On a periodic dividend barrier strategy in the dual model with continuous monitoring of solvency.
\textit{Insurance: Mathematics and Economics}.
\textbf{52}, 98-113.

\bibitem{Avanzi}
Avanzi, B., Gerber, H. U. and E. S. W. Shiu. (2007).
Optimal dividends in the dual model.
\textit{Insurance: Mathematics and Economics}.
\textbf{41}, 111-123.

\bibitem{BE}
Bayraktar, E. and M. Egami. (2008).
Optimizing venture capital investment in a jump diffusion model.
\textit{Mathematical Methods of Operations Research}.
\textbf{67}, 21-42.

\bibitem{Blind}
Blind, K. (2012). 
The influence of regulations on innovation: A quantitative assessment for OECD countries.
\textit{Research Policy}.
\textbf{41}, 391-400.

\bibitem{Boogaert}
Boogaert, P. and J. Haezendonck. (1989).
Delay in claim settlement. 
\textit{Insurance: Mathematics and Economics}.
\textbf{8}, 321-330.

\bibitem{Braeutigam}
Braeutigam, R. R. (1979).
The effect of uncertainty in regulatory delay on the rate of innovation.
\textit{Law and Contemporary Problems}.
\textbf{43}, 98-111.

\bibitem{Bremaud}
Br\'{e}maud, P. (2000).
An intensity property of Lundberg's estimate for delayed claims.
\textit{Journal of Applied Probability}.
\textbf{37}, 914-917.

\bibitem{CheungI}
Cheung, E. C. K. (2012).
A unifying approach to the analysis of business with random gains.
\textit{Scandinavian Actuarial Journal}.
\textbf{2012}, 153-182.

\bibitem{CheungII}
Cheung, E. C. K. and S. Drekic. (2008).
Dividend moments in the dual risk model: exact and approximate approaches.
\textit{ASTIN Bulletin}.
\textbf{38}, 399-422.

\bibitem{DassiosZhao}
Dassios, A. and H. Zhao. (2013).
A risk model with delayed claims. 
\textit{Journal of Applied Probability}.
\textbf{50}, 686-702.

\bibitem{DeFinetti}
De Finetti, B. (1957).
Su un' impostazione alternativa dell teoria collettiva del rischio. 
In: \textit{Transactions of the XVth International Congress of Actuaries}. 
\textbf{2}, pp. 433-443.

\bibitem{DeVylder}
De Vylder, F. E. and M. J. Goovaerts. (1998).
Discussion of `The time value of ruin' by Gerber and Shiu.
\textit{North American Actuarial Journal}.
\textbf{2}, 72-74.

\bibitem{FZ}
Fahim, A. and L. Zhu (2015).
Optimal investment in a dual risk model. 
\textit{arXiv:1510.04924}.

\bibitem{Gerber}
Gerber, H. U. (1979).
\textit{An Introduction to Mathematical Risk Theory}.
S. S. Hu\'{e}bner Foundation Monograph, Series No. 8.

\bibitem{Goulden}
Goulden, I. P. and D. M. Jackson. (1983).
\textit{Combinatorial Enumeration}.
Wiley, New York.

\bibitem{Kluppelberg}
Kl\"{u}ppelberg, C. and T. Mikosch. (1995).
Explosive Poisson shot noise processes with applications to risk reserves.
\textit{Bernoulli}.
\textbf{1}, 125-147.

\bibitem{Lin}
Lin, X. S. and G. E. Willmot. (1999).
Analysis of a defective renewal equation arising in ruin theory.
\textit{Insurance: Mathematics and Economics}.
\textbf{25}, 63-84.

\bibitem{Macci}
Macci, C. and G. L. Torrisi. (2004).
Asymptotic results for perturbed risk processes with delayed claims.
\textit{Insurance: Mathematics and Economics}.
\textbf{34}, 307-320.

\bibitem{Mirasol}
Mirasol, N. M. (1963).
The output of an $M/G/\infty$ queuing system is Poisson. 
\textit{Operations Research}.
\textbf{11}, 282-284.

\bibitem{Ng}
Ng, A. C. Y. (2009).
On a dual model with a dividend threshold.
\textit{Insurance: Mathematics and Economics}.
\textbf{44}, 315-324.

\bibitem{NgII}
Ng, A. C. Y. (2010). 
On the upcrossing and downcrossing probabilities of a
dual risk model with phase-type gains.
\textit{ASTIN Bulletin}
\textbf{40}, 281-306.

\bibitem{Prabhu}
Prabhu, N. U. (1961).
On the ruin problem of collective risk theory.
\textit{Ann. Math. Statist.}
\textbf{32}, 757-764.

\bibitem{Prieger}
Prieger, J. E. (2007).
Regulatory delay and the timing of product innovation.
\textit{International Journal of Industrial Organization}.
\textbf{25}, 219-236.

\bibitem{PriegerII}
Prieger, J. E. (2008).
Product innovation, signaling, and endogenous regulatory delay.
\textit{Journal of Regulatory Economics}.
\textbf{34}, 95-118.

\bibitem{RCE}
Rodr\'{i}guez, E., Cardoso, R. M. R. and A. D. Eg\'{i}dio dos Reis. (2015).
Some advances on the Erlang(n) dual risk model.
\textit{ASTIN Bulletin}.
\textbf{45}, 127-150.

\bibitem{Trufin}
Trufin, J., Albrecher, H. and M. Denuit. (2011).
Ruin problems under IBNR dynamics.
\textit{Appl. Stoch. Models Bus. Ind.}
\textbf{27}, 619-632.

\bibitem{Waters}
Waters, H. R. and A. Papatriandafylou. (1985).
Ruin probabilities allowing for delay in claims settlement.
\textit{Insurance: Mathematics and Economics}.
\textbf{4}, 113-122.

\bibitem{YS}
Yang, C., and K. P. Sendova. (2014).
The ruin time under the Sparre-Andersen dual model.
\textit{Insurance: Mathematics and Economics}.
\textbf{54}, 28-40.

\bibitem{YSII}
Yang, C., Sendova, K.P. and Z. Li. (2015).
Parisian ruin under the dual L\'{e}vy risk model. Submitted.

\bibitem{Yuen}
Yuen, K. C., Guo, J. and W. Ng. (2005).
On ultimate ruin in a delayed-claims risk model.
\textit{Journal of Applied Probability}.
\textbf{42}, 163-174.

\bibitem{Zhu}
Zhu, L. (2015).
A state-dependent dual risk model.
\textit{arXiv: 1510.03920}.
\end{thebibliography}
\end{document}